\documentclass[conference,a4paper]{IEEEtran}
\IEEEoverridecommandlockouts
\usepackage[utf8]{inputenc} 
\usepackage[T1]{fontenc}
\usepackage{url}
\usepackage{ifthen}
\usepackage[cmex10]{amsmath}
\usepackage{array} 

\usepackage{enumerate}
\usepackage{amssymb}
\usepackage{amsmath} 
\usepackage[lined,boxed,commentsnumbered, ruled]{algorithm2e}
\usepackage{mathrsfs}
\usepackage{bm}
\usepackage{tikz}
\usepackage{algorithmic}
\usetikzlibrary{arrows}
\usepackage{subfigure}
\usepackage{graphicx,booktabs,multirow}
\usepackage{epstopdf}
\usepackage{subfigure}
\usepackage{multirow}
\usepackage{tabularx} 
\usepackage{booktabs}

\definecolor{colorhkust}{RGB}{20,43,140}
\definecolor{colortsinghua}{RGB}{116,52,129}
\definecolor{color1}{RGB}{128,0,0}

\usepackage{amsthm}
\usepackage{verbatim}
\usepackage{color}

\newtheorem{lemma}{Lemma}

\usepackage[numbers,sort&compress]{natbib}

\begin{document}

        \title{Coordinated Passive Beamforming for Distributed Intelligent Reflecting Surfaces Network}
        
\author{\IEEEauthorblockN{Jinglian He$^{*\dag\S}$, Kaiqiang Yu$^*$,  and Yuanming Shi$^*$
               }\\      
        \IEEEauthorblockA{
                $^{*}$School of Information Science and Technology,
                ShanghaiTech University, Shanghai 201210, China\\
                $^{\dag}$Shanghai Institute of Microsystem and Information Technology, Chinese Academy of Sciences, China\\
                $^{\S}$University of Chinese Academy of Sciences, Beijing 100049, China\\
                E-mail: hejl1@shanghaitech.edu.cn, yukaiqiangsdu@gmail.com, shiym@shanghaitech.edu.cn\\
        }
    }
        \maketitle
\begin{abstract}
Intelligent reflecting surface (IRS) is a proposing technology in 6G to enhance the performance of wireless networks by smartly reconfiguring the propagation environment with a large number of passive reflecting elements. However, current works mainly focus on single IRS-empowered wireless networks, where the channel rank deficiency problem has emerged. In this paper, we  propose a distributed IRS-empowered communication network architecture, where multiple source-destination pairs communicate through multiple distributed IRSs. We further contribute to maximize the achievable sum-rates in this network  via jointly optimizing the transmit power vector at the sources and the phase shift matrix with passive beamforming at all distributed IRSs. Unfortunately, this problem turns out to be non-convex and highly intractable, for which an alternating approach is developed via solving the resulting fractional programming problems alternatively. In particular, the closed-form expressions are proposed for coordinated passive beamforming at IRSs. The numerical results will demonstrate the  algorithmic advantages and desirable performances of the distributed IRS-empowered communication network.

\end{abstract}

\section{Introduction}

The 6G network is envisioned to support ubiquitous intelligent services with challenging requirements on data rates, latency, and connectivity, for which an AI empowered network architecture, i.e., network intelligentization, subnetwork evolution and intelligent radio, is embraced \cite{6g}. Inspired by these trends, a novel communication paradigm of ''smart radio environment'' enabled by intelligent reflecting surfaces (IRSs) has been proposed to enhance the spectrum and energy efficiency of the wireless networks by reconfiguring the signal propagation environments \cite{smart,1906-09490,ch}. Specifically, IRS is a reconfigurable metasurface of electromagnetic (EM) material consisting of  massive passive reflecting elements, each of which is capable of independently reflecting the incident signal by changing its amplitude and phase \cite{8466374,1906-06578}. By smartly adjusting the reflected signal propagation, we can constructively integrate reflected signals with non-reflected ones and destructively cancel the interference at receiver, thereby achieving the desired performance of wireless networks \cite{1905-00152}. 

Recently, there have been significant progresses on beamforming design for IRS-empowered wireless network \cite{8647620,1906-09434,shen11,1904-12475}. The transmit power minimization problem was exploited via jointly optimizing active beamforming at the base station (BS) and passive beamforming at the IRS in the proposed IRS-empowered MISO wireless system \cite{8647620} and non-orthogonal multiple access (NOMA) \cite{1906-09434}. The weighted sum-rate maximization problem was further considered in \cite{shen11} by fractional programming. Moreover, IRS was deployed in multiple access networks to boost the received signal power for over-the-air computation \cite{1904-12475}.

However, existing works mainly focused on the wireless networks with a single IRS, where the channel rank deficiency problem and computation problem have emerged \cite{1906-02360}. Specifically, the rank-one line-of-sight (LoS) channel is often considered to model the BS-to-IRS link, which yields the rank deficiency problem and limits the system capacity to serve multi-users. In addition, IRS is individually designed according to local channel information, thereby inducing channel estimation and computation tasks for optimization on local IRS controller. One promising solution is to deploy distributed IRSs in wireless networks, where BS-to-IRS channel can be generated as the sum of multiple rank-one channels, thereby guaranteeing high rank channels \cite{1906-02360}. Moreover, in the distributed IRS-empowered communication network architecture, distributed IRSs are coordinated via a central network controller, which reduces the computation at IRSs \cite{1906-09490}.

With the benefits of centralized computation resources and coordinated passive beamforming for phase
shift matrix design at IRSs, the proposed distributed IRS-empowered wireless system can significantly improve network spectral efficiency and energy efficiency in dense wireless networks \cite{DBLP:journals/cm/ShiZCL18}. In this paper, we shall propose to deploy distributed IRSs in distributed  wireless system with multiple source-destination pairs. Our goal is to maximize the sum-rate via jointly optimizing the transmit power vector at the sources and the phase shift matrix with passive beamforming at all distributed IRSs. The formulated problem turns out to be non-convex and highly intractable due to the multi-user interference and joint optimization. We thus design an iterative alternating algorithm by decoupling the optimization variables, which divides the original problem into two tractable subproblems, i.e., power control problem at sources and coordinated passive beamforming problem at IRSs.  The resulting multiple-ratio fractional programming subproblem can be reformulated as a biconvex problem via quadratic transform  \cite{shen11}, supported by an alternating convex
search approach to solve it \cite{Gorski2007}. 

\section{System Model and problem formulation}
\subsection{System Model}

We consider a distributed IRS-empowered communication system consisting of $K$ single-antenna source-destination pairs and $L$ cooperative intelligent reflecting surfaces (IRS)s controlled by a coordinated IRS controller. The $l$-th IRS has $M_l$ passive elements that are used for assisting the communication from source to destination via dynamically adjusting the phase shift according to the channel state information (CSI). In particular, IRS can operate in two coordinated modes, i.e., receiving mode for sensing environment and reflecting mode for scattering the incident signals from the sources \cite{8647620,IET}. In previous single IRS-empowered systems, IRS individually changes the phase shift based on the local channel information estimated by its controller, thereby yielding the channel estimation and computation tasks for optimization on single IRS. However, in our distributed IRS-empowered system, distributed IRSs are coordinated by a central IRS controller according to the global channel information received from each single IRS, thus avoiding the computation on each single IRS.  
 Due to the magnitude path loss, we ignore the power of signals reflected by IRS twice or more times \cite{8647620}. Furthermore, we assume a quasi-static flat-fading channel model with prefect CSI for all channels. 
 Thus, the system consists of three components, i.e., source-IRS link, IRS-destination link, source-destination link. To further simplify, we assume that the direct link does not exist, which represents it is either blocked or has negligible receive power. With this assumption, the signals transmitted from the sources to the destinations experience two phases, as shown in Fig. \ref{model}. In the first phase, the sources transmit signals to the IRS. The received signal $\bm{t}_l\in\mathbb{C}^{M_l\times 1}$ at the $l$-th IRS is given by
 \setlength\arraycolsep{2pt}
\begin{eqnarray}
\bm{t}_l=\sum_{k=1}^K \sqrt{p_k}\bm{h}_{l,k}s_k+\bm{n}_{l,k},
\end{eqnarray}
where $\bm{n}_{l,k}$ is the Gaussian noise with distribution $\mathcal{C}\mathcal{N}(0,\sigma_r^2)$, $s_k$ denotes the transmitted symbol from $k$-th sources with $\mathbb{E}\{|s_k|^2\}=1$, $p_k$ represents the transmit power of the $k$-th sources with power constraint $p_k\leq P_{\max}$, and $\bm{h}_{l,k}\in \mathbb{C}^{M_l\times 1}$ is the vector containing the channel coefficients from the $k$-th source to the $l$-th IRS. In the second phase, the $l$-th IRS reflects the received signal based on the diagonal phase shifts matrix $\bm{\Theta}_l=\beta_l\text{diag}( e^{j\phi_{l,1}},..., e^{j\phi_{l,M_l}})$ with $\phi_m\in[0,2\pi]$ and $\beta_l\in [0,1]$ as the amplitude reflection coefficient on the incident signals. Without loss of generality, we assume $\beta_l=1$. Therefore, the reflected signal at the $l$-th IRS can be written as
\begin{eqnarray}
\bm{r}_l=\sum_{k=1}^K \sqrt{p_k}\bm{\Theta}_l\bm{h}_{l,k}s_k+\bm{\Theta}_l\bm{n}_{l,k}.
\end{eqnarray}
Then, the signal received at the $k$-th destination is
\begin{eqnarray}
\label{s-d-channel}
        y_k=\sum_{l=1}^L\sqrt{p_k}\bm{g}_{k,l}^{\sf{H}}\bm{\Theta}_l\bm{h}_{l,k}s_k+\sum_{i\neq k}\sum_{l=1}^{L}\sqrt{p_i}\bm{g}_{k,l}^{\sf{H}}\bm{\Theta}_l\bm{h}_{l,i}s_i\nonumber\\
        +\sum_{l=1}^L\bm{g}_{k,l}^{\sf{H}}\bm{\Theta}_l\bm{n}_l+w_k,\ \ \ \ \ \ \ \ \ \ \ \ \ \ \ \ \ \ \ \ \ \ \ \ \ \ \ \ \ \ \ 
\end{eqnarray}
where $w_k$ is the additive noise at the $k$-th destination with distribution $\mathcal{C}\mathcal{N}(0,\sigma_d^2)$, $\bm{g}_{k,l}\in \mathbb{C}^{M_l\times 1}$ denotes the channel coefficients between the $l$-th IRS and the $k$-th destination. Let $\bm{\theta}=[e^{j\phi_{1,1}},...,e^{j\phi_{1,M_1}},e^{j\phi_{2,1}},...,e^{j\phi_{L,M_L}}]^{\sf{H}}$, $\bm{v}_{i,k}=[\bm{h}_{1,i}^{\sf{T}}\text{diag}(\bm{g}_{k,1}^{\sf{H}}),...,\bm{h}_{L,i}^{\sf{T}}\text{diag}(\bm{g}_{k,L}^{\sf{H}})]^{\sf{T}}$, and $\bm{z}_k=[\bm{n}_1^{\sf{T}}\text{diag}(\bm{g}_{k,1}^{\sf{H}}),...,\bm{n}_L^{\sf{T}}\text{diag}(\bm{g}_{k,L}^{\sf{H}})]^{\sf{T}}$. We can rewrite (\ref{s-d-channel}) as
\begin{eqnarray}
        \label{new-s-d-channel}
        y_k=\sqrt{p_k}\bm{\theta}^{\sf{H}}\bm{v}_{k,k}s_k+\sum_{i\neq k}\sqrt{p_i}\bm{\theta}^{\sf{H}}\bm{v}_{i,k}s_i
        +\bm{\theta}^{\sf{H}}\bm{z}_k+w_k.
\end{eqnarray}

Based on the single-user detection, each destination treats the interferences as Gaussian noise. Therefore, the signal-to-interference-plus-noise ratio (SINR) at the $k$-th destination can be written as
\begin{eqnarray}
\text{SINR}_k=\frac{p_k|\bm{\theta}^{\sf{H}}\bm{v}_{k,k}|^2}{\sum_{i\neq k}p_i|\bm{\theta}^{\sf{H}}\bm{v}_{i,k}|^2+|\bm{\theta}^{\sf{H}}\bm{z}_k|^2+\sigma_d^2}.
\end{eqnarray}
Then, the achievable rate of the $k$-th destination is given by
\begin{eqnarray}
        R_k(\bm{p},\bm{\theta})=\frac{1}{2}\log_2(1+\text{SINR}_k),
\end{eqnarray}
where $\bm{p}=[p_1,...,p_K]^{\sf{T}}$.
\begin{figure}[t]
	\center
	\includegraphics[scale = 0.3]{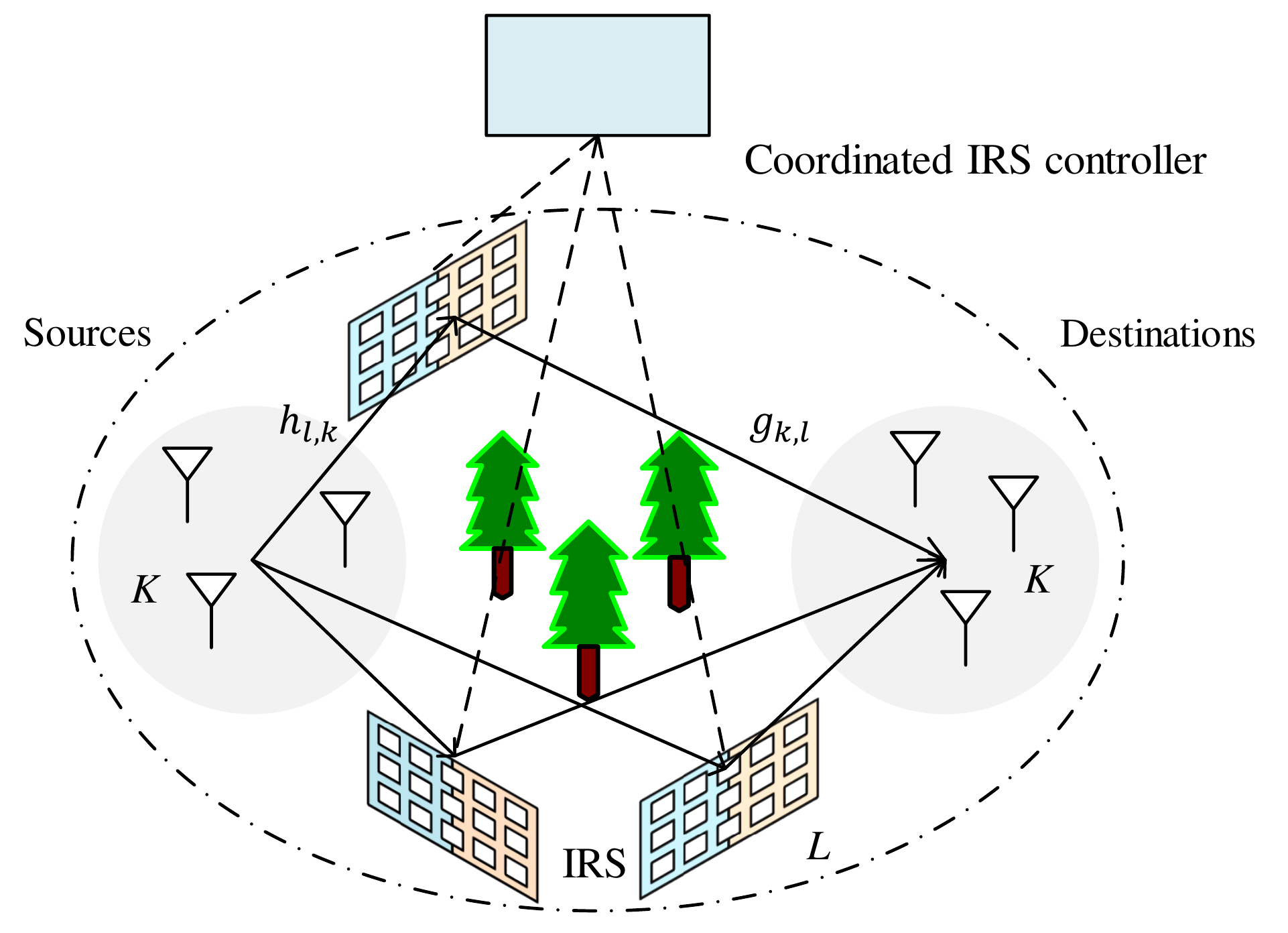}
	\caption{An interference channel with distributed intelligent reflecting surfaces.}
	\label{model}
\end{figure}

\subsection{Problem Formulation}
In this paper, we aim to maximize the sum-rate by jointly designing the power coefficient vector $\bm{p}$ and the reflection coefficient vector $\bm{\theta}$ under the following power transmit constraint $\mathcal{P}$ and the IRS phase constraint $\mathcal{C}$:
\begin{eqnarray}
\mathcal{P}=\{p\big|0\leq p\leq P_{\max}\},\ \mathcal{C}=\{\theta\big||\theta|^2\leq 1\}.
\end{eqnarray}
Let $N=\sum_{l=1}^L M_l$ denote the number of passive elements. The sum-rate maximization problem can be formulated as
\begin{eqnarray}
\mathscr{P}:\mathop {\rm{maximize}}_{\bm{p},\bm{\theta}}\ \ R_{\text{sum}}=\sum_{k=1}^K\frac{1}{2}\log_2(1+\text{SINR}_k)\ \ \ \ \ \nonumber\\
{\rm{subject~to}}\ \ \ p_k\in\mathcal{P},\ \forall k=1,2,...,K,\ \ \ \ \ \ \ \ \ \ \ \nonumber\\
\theta_i\in\mathcal{C},\ \forall i=1,2,...,N.\ \ \ \ \ \ \ \ \ \ \ \ \ 
\end{eqnarray}


This problem turns out to be non-convex due to the multi-user interference and ratio operation. Another challenge introduced by jointly optimizing $\bm{p}$ and $\bm{\theta}$ makes problem $\mathscr{P}$ highly intractable. In next section, we propose to design an efficient iterative algorithm to find a suboptimal solution via alternatively optimizing $\bm{p}$ and $\bm{\theta}$.

\section{Optimal Coordinated Passive Beamforming}
In this section, we propose a low-complexity iterative suboptimal algorithm, which divides problem $\mathscr{P}$ into several tractable subproblems via alternatively optimizing $\bm{p}$ and $\bm{\theta}$.

\subsection{Lagrangian Dual Reformulation}
Basically, we can directly employ fractional programming to solve the original problem via applying the quadratic transform to each SINR term. In this approach, convex optimaizations problem need to be solved numerically in each iteration, which incurs large computation. Another more desirable and efficient method, based on a Lagangian dual reformulation of the original problem, performs in closed form at each iteration. Although the original problem is non-convex, we can always derive its upper bound from Lagrangian dual form to find a suboptimal solution.
To develop this efficient algorithm, we consider the following Lagrangian dual reformulation of the original problem $\mathscr{P}$ proposed in \cite{shen11}:
\begin{eqnarray}
\ \ \ \ \ \ \ \ \ \ \mathscr{L}:\mathop {\rm{maximize}}_{\bm{p},\bm{\theta},\bm{\mu}}\ \  f_1(\bm{p},\bm{\theta},\bm{\mu})\ \ \ \ \ \ \ \ \ \ \ \ \ \ \ \ \ \ \ \ \  \nonumber\\
{\rm{subject~to}}\ \ \ p_k\in\mathcal{P},\ \forall k=1,2,...,K,\ \ \ \nonumber\\
\theta_i\in\mathcal{C},\ \forall i=1,2,...,N.\ \ \ \ \ 
\end{eqnarray}
where $\bm{\mu}=[\mu_1,..,\mu_K]^{\sf{T}}$ denotes an auxiliary variable vector, and the new objective is given by
\begin{eqnarray}
f_1(\bm{p},\bm{\theta},\bm{\mu})=\sum_{k=1}^K \frac{1}{2}\log_2(1+\mu_k)-\sum_{k=1}^K\frac{1}{2}\mu_k\nonumber\\
+\sum_{k=1}^K\frac{(1+\mu_k)p_k|\bm{\theta}^{\sf{H}}\bm{v}_{k,k}|^2}{2(\sum_{i= 1}^Kp_i|\bm{\theta}^{\sf{H}}\bm{v}_{i,k}|^2+|\bm{\theta}^{\sf{H}}\bm{z}_k|^2+\sigma_d^2)}.
\end{eqnarray}
Based on the proposed iterative algorithm, $\bm{p}$ and $\bm{\theta}$ can be fixed firstly. Then, optimal $\mu_k$ can be obtained via setting $\partial f_1/\partial \mu_k$ to be zero, i.e.,
\begin{eqnarray}
\label{optimal_mu}
\mu_k^{*}=\frac{p_k|\bm{\theta}^{\sf{H}}\bm{v}_{k,k}|^2}{\sum_{i\neq k}p_i|\bm{\theta}^{\sf{H}}\bm{v}_{i,k}|^2+|\bm{\theta}^{\sf{H}}\bm{z}_k|^2+\sigma_d^2}\ k=1,...,K.
\end{eqnarray}
After updating $\mu_k$ as $\mu_k^*$, the first two terms of $f_1$ remain constant. Therefore, problem $\mathscr{L}$ can be further reduced to
\begin{eqnarray}
\ \ \ \ \ \ \ \ \ \ \mathscr{L}_1:\mathop {\rm{maximize}}_{\bm{p},\bm{\theta},\bm{\mu}}\ \  f_2(\bm{p},\bm{\theta})\ \ \ \ \ \ \ \ \ \ \ \ \ \ \  \ \ \ \ \ \ \ \ \  \nonumber\\
{\rm{subject~to}}\ \ \ p_k\in\mathcal{P},\ \forall k=1,2,...,K,\ \ \ \nonumber\\
\theta_i\in\mathcal{C},\ \forall i=1,2,...,N.\ \ \ \ \ 
\end{eqnarray}
where the objective function of $\mathscr{L}_1$ is defined by
\begin{eqnarray}
f_2(\bm{p},\bm{\theta})=\sum_{k=1}^K\frac{(1+\mu_k)p_k|\bm{\theta}^{\sf{H}}\bm{v}_{k,k}|^2}{2(\sum_{i= 1}^Kp_i|\bm{\theta}^{\sf{H}}\bm{v}_{i,k}|^2+|\bm{\theta}^{\sf{H}}\bm{z}_k|^2+\sigma_d^2)}.
\end{eqnarray}

Problem $\mathscr{L}_1$ is still non-convex due to the sum of multiple-ratio form, which can be solved by fractional programming framework proposed in \cite{shen11}. Another challenge induced by joint optimization need to be tackled via alternatively optimizing $\bm{p}$ and $\bm{\theta}$ in the following subsections. Specifically, in each iteration, we first update $\bm{\mu}$, then optimize $\bm{p}$ and $\bm{\theta}$ respectively. Repeat this procedure until $R_{\text{sum}}$ converges.

\subsection{Power Control}
In this subsection, we shall optimize $\bm{p}$ based on the fixed $\bm{\mu}$ and $\bm{\theta}$. Problem $\mathscr{L}_1$ can be recast as the following power control sub-problem with given $\bm{\theta}$:
\begin{eqnarray}
\mathop {\rm{maximize}}_{\bm{p}}\ f_2(\bm{p})\ \ 
{\rm{subject~to}}\ \ p_k\in\mathcal{P},\ \forall k=1,2,...,K.\ \ 
\end{eqnarray}
The classic power control problem, as a multiple-ratio fraction programming problem, has been well exploited by a novel quadratic transform in \cite{shen11}. Based on its FP framework, above power control problem can be further equivalently reformulate as following biconvex optimization problem:
\begin{eqnarray}
\mathop {\rm{maximize}}_{\bm{p},\bm{\alpha}}\ g_1(\bm{p},\bm{\alpha})\ \ 
{\rm{subject~to}}\ p_k\in\mathcal{P},\ \forall k=1,2,...,K.\ \ 
\end{eqnarray}
where $\bm{\alpha}\in \mathbb{R}^{K\times 1}$ denotes the auxiliary variable vector and the objective function is defined by
\begin{eqnarray}
g_1(\bm{p},\bm{\alpha})=\sum_{k=1}^K \alpha_k\sqrt{2(1+\mu_k)|\bm{\theta}^{\sf{H}}\bm{v}_{k,k}|^2p_k}\nonumber\\
-\sum_{k=1}^{K}\alpha_k^2(\sum_{i=1}^Kp_i|\bm{\theta}^{\sf{H}}\bm{v}_{i,k}|^2+|\bm{\theta}^{\sf{H}}\bm{z}_k|^2+\sigma_d^2).
\end{eqnarray}
By utilizing the convex substructures of above problem, an alternating convex search approach can be further developed to obtain a local optimal solution \cite{Gorski2007}. In particular, only one variable is optimized at each step while others are fixed. Given fixed $\bm{p}$, the optimal $\bm{\alpha}$ is
\begin{eqnarray}
\label{optimal_alpha}
        \alpha_k^*=\frac{\sqrt{2(1+\mu_k)|\bm{\theta}^{\sf{H}}\bm{v}_{k,k}|^2p_k}}{2(\sum_{i=1}^Kp_i|\bm{\theta}^{\sf{H}}\bm{v}_{i,k}|^2+|\bm{\theta}^{\sf{H}}\bm{z}_k|^2+\sigma_d^2)}.
\end{eqnarray}
For fixing $\bm{\alpha}$, the optimal $\bm{p}$ is given by
\begin{eqnarray}
\label{optimal_power}
        p_k^*=\min\left\{P_{\max},\frac{\alpha_k^2(1+\mu_k)|\bm{\theta}^{\sf{H}}\bm{v}_{k,k}|^2}{2(\sum_{i=1}^K\alpha_i^2|\bm{\theta}^{\sf{H}}\bm{v}_{k,i}|^2)^2}\right\}.
\end{eqnarray}
\subsection{Optimizing Reflection Coefficients}
In this subsection, we propose to solve another fractional programming sub-problem via optimizing $\bm{\theta}$ over the fixed $\bm{p}$. Based on the quadratic transform proposed in \cite{shen11}, we can obtain the following equivalent problem:
\begin{eqnarray}
\mathop {\rm{maximize}}_{\bm{\theta},\bm{\beta}}\ g_2(\bm{\theta},\bm{\beta})\ \ 
{\rm{subject~to}}\ \theta_i\in\mathcal{C},\ \forall i=1,...,N.
\end{eqnarray}
where $\bm{\beta}\in\mathbb{C}^{K\times 1}$ denotes the auxiliary variable vector and the new objective function is given by
\begin{eqnarray}
        g_2(\bm{\theta},\bm{\beta})=\sum_{k=1}^K \sqrt{2p_k(1+\mu_k)}{\text{Re}}\{\beta^+_k\bm{\theta}^{\sf{H}}\bm{v}_{k,k}\}\nonumber\\
        -\sum_{k=1}^K|\beta_k|^2(\sum_{i=1}^Kp_i|\bm{\theta}^{\sf{H}}\bm{v}_{i,k}|^2+|\bm{\theta}^{\sf{H}}\bm{z}_k|^2+\sigma_d^2),
\end{eqnarray}
where $\beta_k^+$ denotes its conjugate.
We propose to solve above problem via alternatively optimizing $\bm{\theta}$ and $\bm{\beta}$. For fixed $\bm{\theta}$, the optimal $\beta_k$ can be obtained by setting $\partial g_2/\partial \beta_k$ to be zero, i.e.,
\begin{eqnarray}
\label{optimal_beta}
        \beta_k^*=\frac{\sqrt{2p_k(1+\mu_k)}\bm{\theta}^{\sf{H}}\bm{v}_{k,k}}{2(\sum_{i=1}^Kp_i|\bm{\theta}^{\sf{H}}\bm{v}_{i,k}|^2+|\bm{\theta}^{\sf{H}}\bm{z}_k|^2+\sigma_d^2)}.
\end{eqnarray}
Consider the phase constraint $|\theta_n|\leq 1$, it can be further rewritten as
\begin{eqnarray}
        \bm{\theta}^{\sf{H}}\bm{e}_i\bm{e}_i^{\sf{H}}\bm{\theta}\leq 1,\ i=1,2,...,N,
\end{eqnarray}
where $\bm{e}_i\in\mathbb{R}^{N\times 1}$ denotes an elementary vector with a one at $i$-th position.
Then, optimizing $\bm{\theta}$ based on the fixed $\bm{\beta}$ is a convex QCQP problem, which can be recast as the following equivalent dual problem by Lagrange dual decomposition:
\begin{eqnarray}
        \mathop {\rm{minimize}}_{\bm{\lambda}}\ \sup_{\bm{\theta}} \mathcal{L}(\bm{\theta},\bm{\lambda})\ \ \ \ \ \ \ \ \ \ \ \ \ \ \ \ \ \ \nonumber\\
        {\rm{subject~to}}\ \lambda_n\geq 0,\ n=1,2,...,N,\ \ \ \ 
\end{eqnarray}
where $\bm{\lambda}\in\mathbb{R}^{N\times 1}$ represents the dual variable, and $\mathcal{L}(\bm{\theta},\bm{\lambda})$ denotes the dual objective function, which is given by
\begin{eqnarray}
        \mathcal{L}(\bm{\theta},\bm{\lambda})=g_2(\bm{\theta},\bm{\beta}^*)-\sum_{i=1}^{N}\lambda_i(\bm{\theta}^{\sf{H}}\bm{e}_i\bm{e}_i^{\sf{H}}\bm{\theta}-1),
\end{eqnarray}
where $\bm{\beta}^*$ is the fixed optimal $\bm{\beta}$. $\mathcal{L}(\bm{\theta},\bm{\lambda})$ is a concave function with respect to $\bm{\theta}$. Since the Slater's condition is satisfied, the duality gap is zero \cite{convex}. Therefore, the optimal $\bm{\theta}$ can be obtained via setting $\partial \mathcal{L}/\partial \bm{\theta}$ to be zero, i.e.,
\begin{eqnarray}
\label{optimal_theta}
        \bm{\theta}^*=\frac{1}{2}\left(\sum_{i= 1}^{N}\lambda_i^*\bm{e}_i\bm{e}_i^{\sf{H}}+\bm{A}\right)^{-1}\bm{u},\nonumber\\
        \bm{A}=\sum_{k=1}^K|\beta_k|^2\left(\sum_{i= 1}^Kp_i\bm{v}_{i,k}\bm{v}_{i,k}^{\sf{H}}\right)+\bm{z}_k \bm{z}_k^{\sf{H}},\nonumber\\
        \bm{u}=\sum_{k=1}^K\sqrt{2p_k(1+\mu_k)}\beta_k^*\bm{v}_{k,k},
\end{eqnarray}
where $\bm{\lambda}^*$ denotes the optimal dual variable vector, which can be obtained via ellipsoid method.
\begin{algorithm}
        \label{AMA}
        \caption{Proposed alternating optimization}
        {\textbf{Input}}: initial point $\bm{p}^{(0)}$ and $\bm{\theta}^{(0)}$, threshold $\epsilon> 0$\\
        {\textbf{for $t= 1$ to $1,2,...$}}, {\textbf{do}}\\
        \ \ \ Update $\bm{\mu}^{(t)}$ by (\ref{optimal_mu});\\
        \ \ \ Update $\bm{\alpha}^{(t)}$ by (\ref{optimal_alpha});\\
        \ \ \ Updata power coefficient vector $\bm{p}^{(t)}$ by (\ref{optimal_power});\\
        \ \ \ Update $\bm{\beta}^{(t)}$ by (\ref{optimal_beta});\\
        \ \ \ Update reflection coefficient vector $\bm{\theta}^{(t)}$ by (\ref{optimal_theta});\\
        \ \ \ {\textbf{if}} $f_1^{(t)}-f_1^{(t-1)}\leq \epsilon$\\
        \ \ \ \ \ \ {\textbf{break}};\\
        \ \ \ {\textbf{end if}};\\
        {\textbf{end for}}.\\
        {\textbf{Output}} $\bm{p}^{(t)}$ and $\bm{\theta}^{(t)}$.
\end{algorithm}
\subsection{Proposed Algorithm}
In this subsection, we present the proposed iterative alternating optimization algorithm in Algorithm \ref{AMA}. Specifically, the algorithm starts with two arbitrary initial vectors $\bm{p}^{(0)}$ with $p_i^{(0)}\in\mathcal{P}$ and $\bm{\theta}^{(0)}$ with $\theta_j^{(0)}\in\mathcal{C}$. Then, we iteratively update $\bm{\mu}^{(i)}$ based on the fixed $\{\bm{p}^{(i-1)},\bm{\theta}^{(i-1)}\}$, $\bm{p}^{(i)}$ and $\bm{\theta}^{(i)}$ based on the fractional programming until $f_1(\bm{p},\bm{\theta},\bm{\mu})$ converges.
\begin{lemma}
        The proposed alternating algorithm is guaranteed to converge, with the sum-rate $R_{\text{sum}}$ monotonically nondecreasing after each iteration.
\end{lemma}

\begin{proof}
Please refer to Appendix for details.
\end{proof}
\begin{figure}[t]
	\center
	\includegraphics[scale = 0.3]{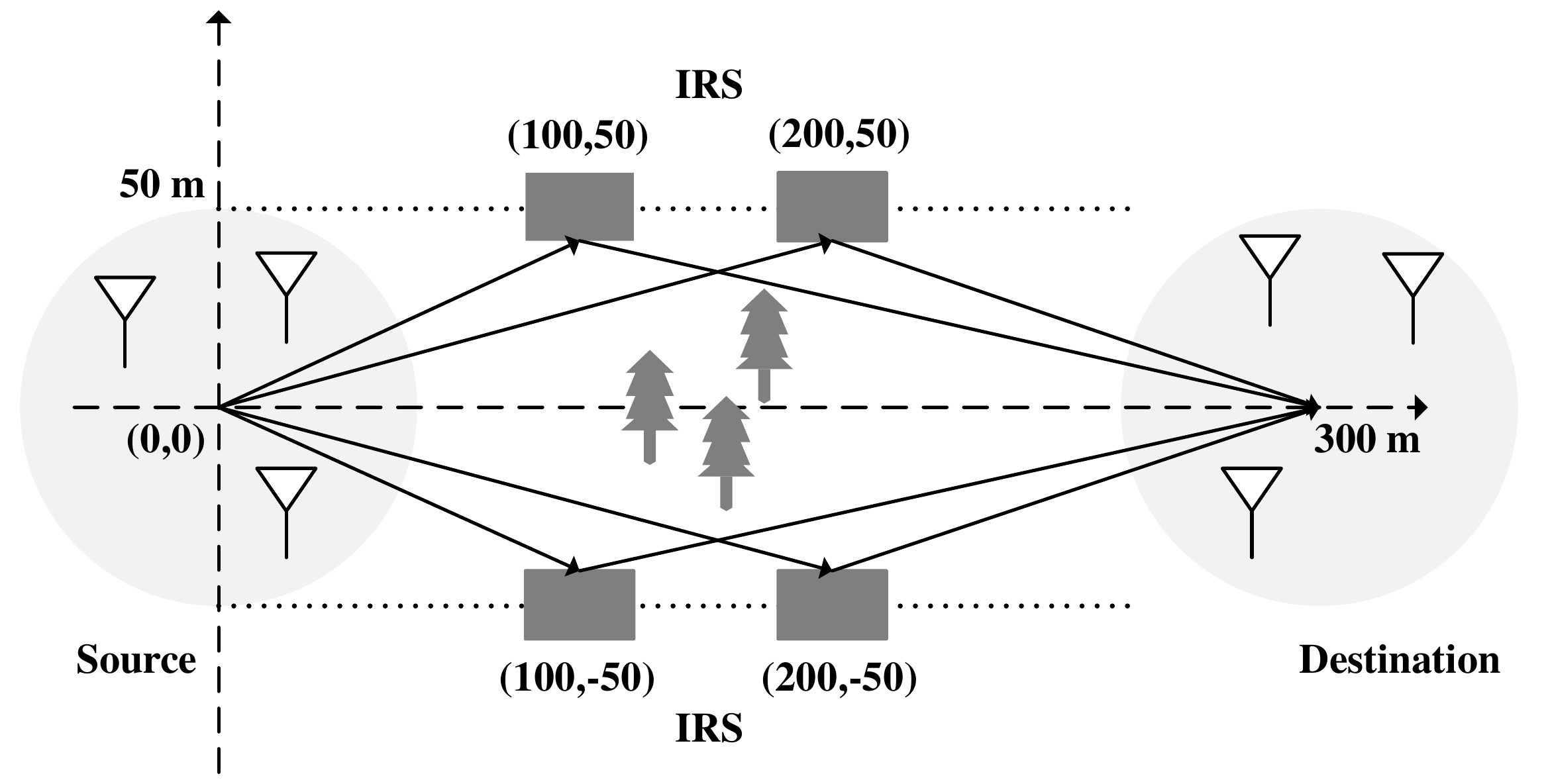}
	\caption{Layout of source-destination pairs and IRS.}
	\label{sim}
\end{figure}

\section{Simulation}
In this section, we simulate the proposed alternating algorithm to evaluate its effectiveness and show the performance of distributed IRS-aided communication systems.

\subsection{Simulation Setting}
We consider a distributed IRS-empowered communication system illustrated in Fig. \ref{sim}, where all single-antenna sources and destinations are uniformly and randomly distributed in two circles centred at $(0,0)$ meters and $(300,0)$ meters with with radius 50 meters, respectively. Four distributed IRSs with a uniform rectangular array of passive reflecting elements are respectively located in $(100,\pm50)$ meters and $(200,\pm50)$ meters. The path loss model we consider is given by $\kappa(d)=T_0(\frac{d}{d_0})^{-\varrho}$,
where $T_0$ denotes the path loss at the reference distance $d_0=1$ meter, $d$ is the link distance and $\varrho$ is the path loss exponent. In this simulation, we assume $T_0=30$ dB, and the path loss for the source-IRS link and the IRS-destination link are respectively set to 2.2 and 2.8. We further assume all the considered channels suffer from Rayleigh fading. To be specific, the channel coefficients are given by
\begin{eqnarray}
        \bm{h}_{l,k}=\sqrt{\kappa(d^{SI}_{l,k})}\bm{\gamma}_{SI},\ \bm{g}_{k,l}=\sqrt{\kappa(d^{ID}_{k,l})}\bm{\gamma}_{ID},
\end{eqnarray}
where $\bm{\gamma}_{SI}\sim\mathcal{CN}(0,\bm{I})$, $\bm{\gamma}_{SI}\sim\mathcal{CN}(0,\bm{I})$, $d^{SI}_{l,k}$ and $d^{ID}_{k,l}$ respectively denote the distance between $k$-th source and $l$-th IRS, the distance between $l$-th IRS and $k$ destination. In addition, we set $\sigma_r^2=\sigma_d^2=0.01$.
\subsection{Simulation Results}

\begin{figure}
	\begin{minipage}[t]{0.5\linewidth}
		\centering
		\includegraphics[scale = 0.28]{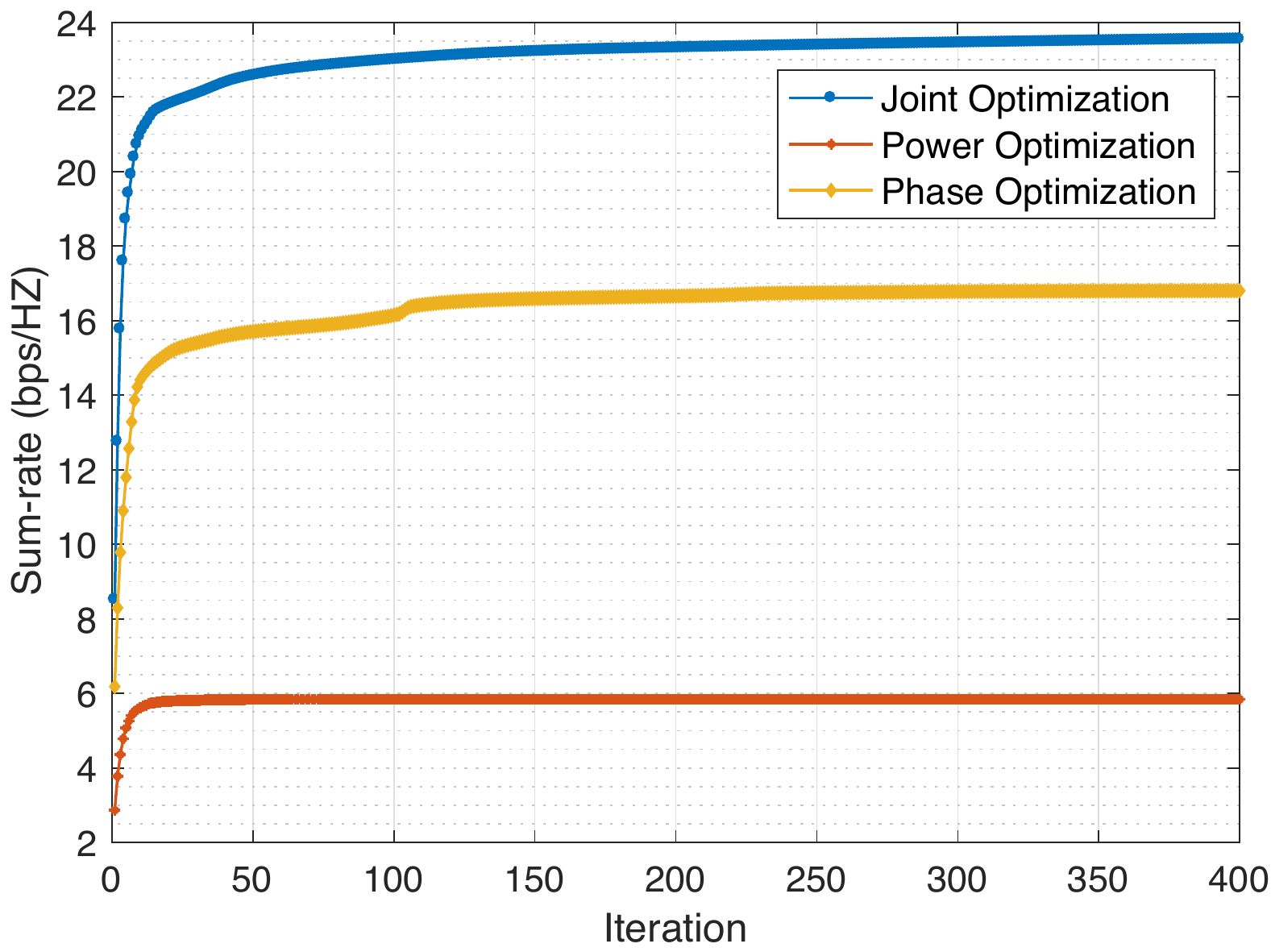}
		\caption{Convergence.}
		\label{final2}
	\end{minipage}%
	\begin{minipage}[t]{0.5\linewidth}
		\centering
		\includegraphics[scale = 0.28]{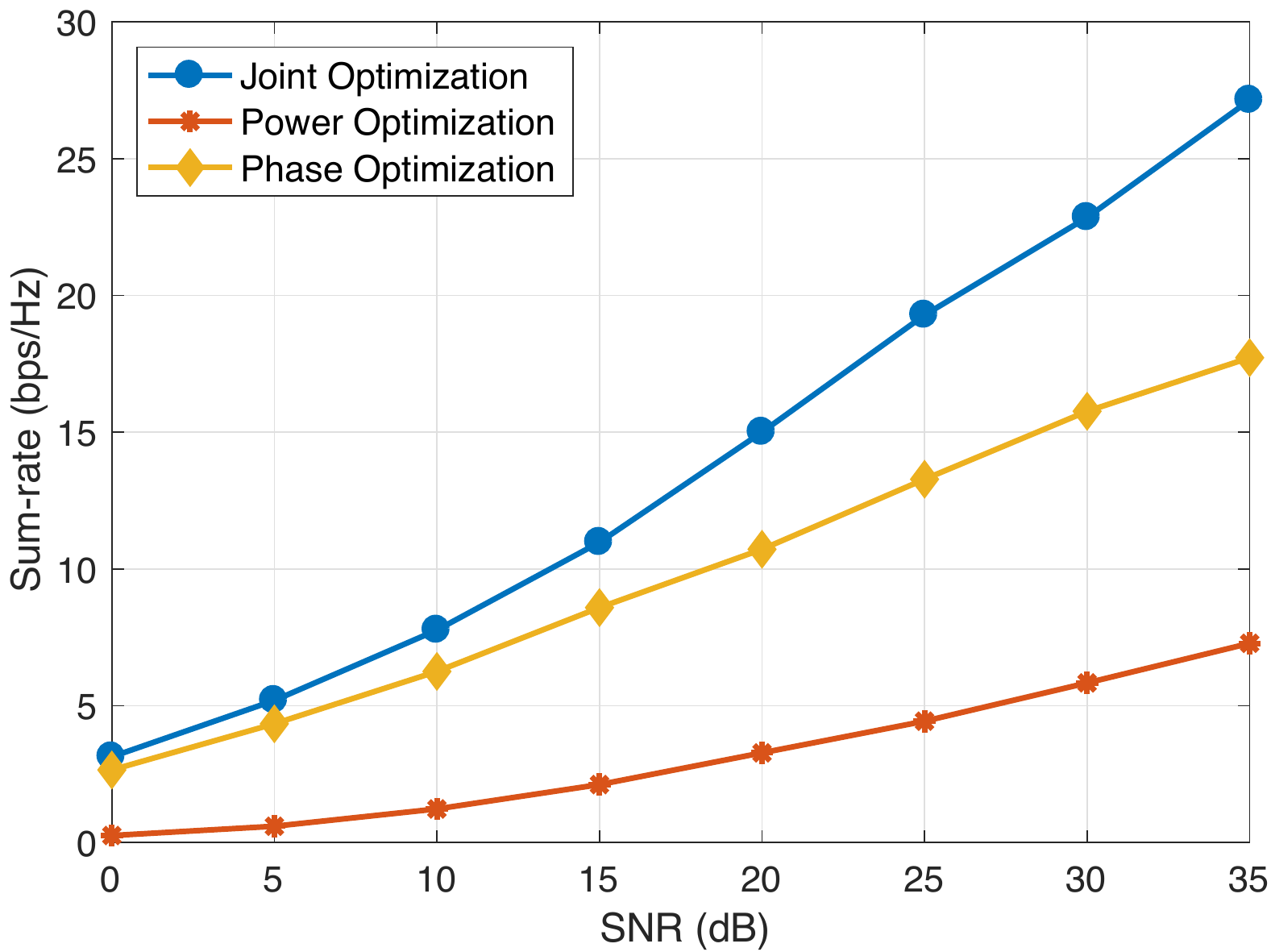}
		\caption{Sum-rate vs. SNR.}
		\label{final1}
	\end{minipage}
\end{figure}


We simulate three different alternating algorithms denoted as \emph{Joint Optimization}, \emph{Power Optimization} and \emph{Phase Optimization}, respectively. 
\begin{itemize}
\item \textbf{Joint Optimization}. Jointly optimizing both $\bm{p}$ and $\bm{\theta}$.
\item \textbf{Power Optimization}. Only optimizing $\bm{p}$ with random phase shift.
\item \textbf{Phase Optimization}. Only optimizing $\bm{\theta}$ under the setting $p_i=P_{\max}\ (1\leq i\leq K)$.
\end{itemize}
All the simulation results are obtained by averaging 100 channel realizations with fixed $K=6$ and $M_l=4\ (1\leq l\leq L)$.

Fig. \ref{final2} demonstrates that our proposed alternating algorithm converges under setting $L=4$ and $\text{SNR}=35$ dB.
We further compare the sum-rate versus different SNR in Fig. \ref{final1} with fixed $L=4$. It can be observed that all the alternating algorithms perform better with the increasing SNR, since the transmit power of the signals reflected by IRS increases. However, joint optimization algorithm achieves higher sum-rate than other two algorithms due to its jointly optimizing transmit power $\bm{p}$ at sources and passive phase shifts $\bm{\theta}$ at IRS. 

Then, we investigate the impact of the number of distributed IRS on sum-rate. Since there is no current work to exploit relevant problems, i.e., optimal distributed IRS positions, we randomly and uniformly deploy $L$ IRSs in the given region $[\pm 100,\pm50]\times [\pm 300, \pm60]$ meters. Fig. \ref{final4} shows the sum-rate increases with the increasing number of distributed IRSs, which demonstrates the admirable performance of proposed distributed IRS-empowered system compared with existing single IRS-empowered wireless networks.


\begin{figure}[t]
        \center
        \includegraphics[scale = 0.305]{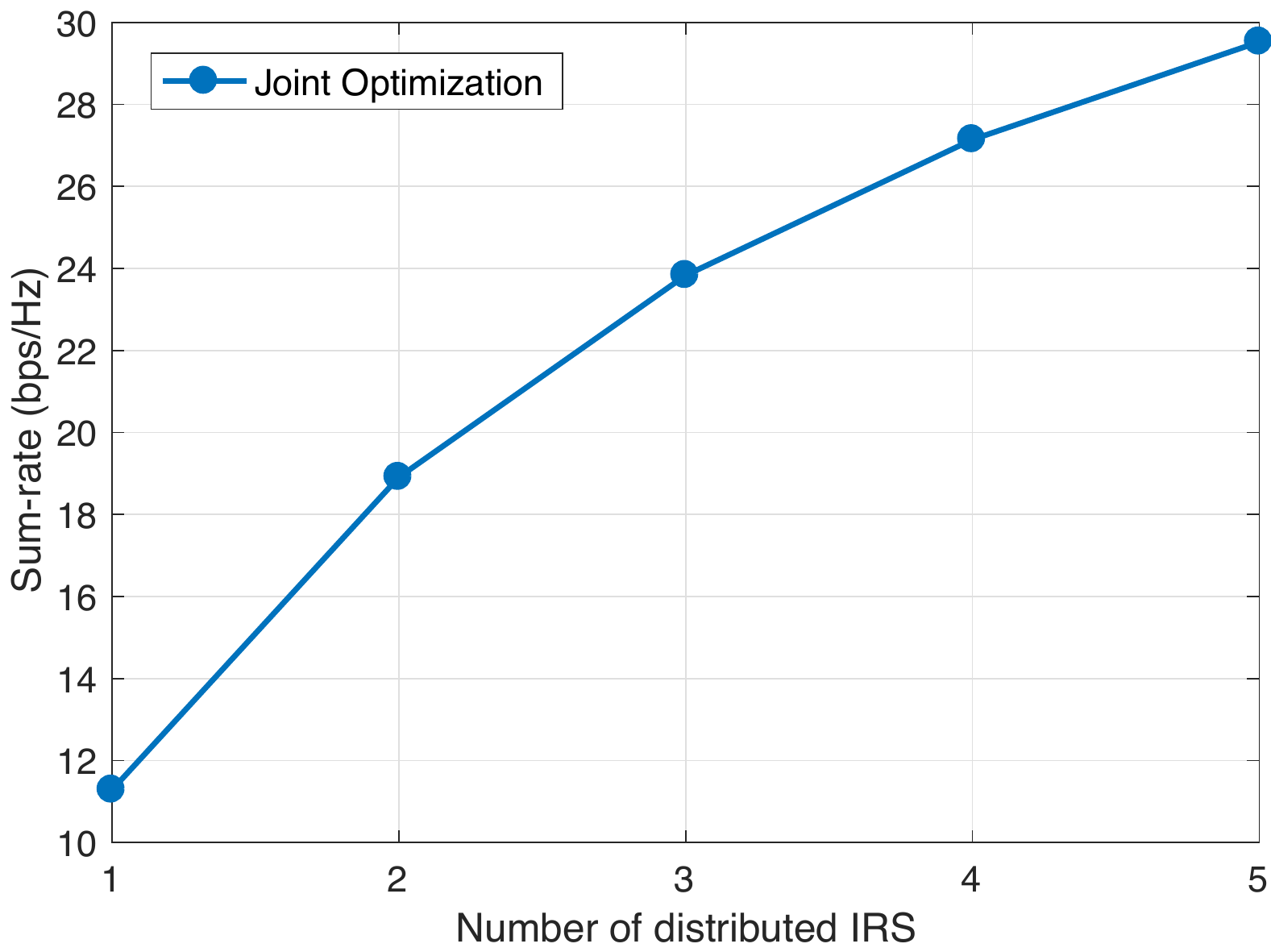}
        \caption{Sum-rate vs. number of distributed IRS.}
        \label{final4}
\end{figure}
\section{Conclusion}

In this paper, we proposed a distributed IRS-empowered wireless network to maximize the achievable sum-rates by jointly optimizing the transmit power vector at the sources and the phase shift matrix with passive beamforming at all distributed IRSs. To solve this non-convex and intractable problem, we presented an alternating algorithm by decoupling transmit power vector and passive beamforming optimization variables, yielding two multiple-ratio fraction programming subproblems. We further transformed the fractional programming problem into biconvex problem, for which an alternating convex search approach with closed-form expressions was developed. Simulation results demonstrated the admirable performance of proposed distributed IRS-empowered system.

\section{appendix}
We now show that our proposed alternating algorithm converges. 
According to the Lagrangian dual reformulation proposed in \cite{shen11}, we have
\begin{eqnarray}
R_{\text{sum}}(\bm{p}^{(t-1)},\bm{\theta}^{(t-1)})=f_1(\bm{p}^{(t-1)},\bm{\theta}^{(t-1)},\bm{\mu}^{(t-1)})\nonumber\\
=\text{constant}(\bm{\mu}^{t-1})+f_2(\bm{p}^{(t-1)},\bm{\theta}^{(t-1)},\bm{\mu}^{(t-1)})\nonumber\\
\leq\text{constant}(\bm{\mu}^{t})+f_2(\bm{p}^{(t-1)},\bm{\theta}^{(t-1)},\bm{\mu}^{(t)}),\ \ \ \ \ \ 
\end{eqnarray}
where $\text{constant}(\bm{\mu})=\sum_{k=1}^K \frac{1}{2}\log_2(1+\mu_k)-\sum_{k=1}^K\frac{1}{2}\mu_k$.
Based on the quadratic transform in \cite{shen11}, we have
\begin{align}
f_2(\bm{p}^{(t-1)},\bm{\theta}^{(t-1)},\bm{\mu}^{(t)})&= g_2(\bm{p}^{(t-1)},\bm{\theta}^{(t-1)},\bm{\alpha}^{(t-1)},\bm{\mu}^{(t)})\nonumber\\
&\leq g_2(\bm{p}^{(t-1)},\bm{\theta}^{(t-1)},\bm{\alpha}^{(t)},\bm{\mu}^{(t)})\nonumber\\
&\leq g_2(\bm{p}^{(t)},\bm{\theta}^{(t-1)},\bm{\alpha}^{(t)},\bm{\mu}^{(t)})\nonumber\\
&=f_2(\bm{p}^{(t)},\bm{\theta}^{(t-1)},\bm{\mu}^{(t)}),
\end{align}
where $\bm{\alpha}^{(t)}$ and $\bm{\theta}^{(t)}$ are updated by (\ref{optimal_alpha}) and (\ref{optimal_power}), respectively. Similarly, we can also obtain
\begin{eqnarray}
        f_2(\bm{p}^{(t)},\bm{\theta}^{(t-1)},\bm{\mu}^{(t)})\leq f_2(\bm{p}^{(t)},\bm{\theta}^{(t)},\bm{\mu}^{(t)}).
\end{eqnarray}
Therefore, we can proof the convergence behavior of proposed alternating algorithm by combining the above equations:
\begin{eqnarray}
f_1(\bm{p}^{(t-1)},\bm{\theta}^{(t-1)},\bm{\mu}^{(t-1)})\leq f_1(\bm{p}^{(t)},\bm{\theta}^{(t)},\bm{\mu}^{(t)}).\nonumber
\end{eqnarray}

\bibliographystyle{IEEEtran}
\bibliography{IEEEabrv,Reference}

\begin{thebibliography}{10}
\providecommand{\url}[1]{#1}
\csname url@samestyle\endcsname
\providecommand{\newblock}{\relax}
\providecommand{\bibinfo}[2]{#2}
\providecommand{\BIBentrySTDinterwordspacing}{\spaceskip=0pt\relax}
\providecommand{\BIBentryALTinterwordstretchfactor}{4}
\providecommand{\BIBentryALTinterwordspacing}{\spaceskip=\fontdimen2\font plus
\BIBentryALTinterwordstretchfactor\fontdimen3\font minus
  \fontdimen4\font\relax}
\providecommand{\BIBforeignlanguage}[2]{{%
\expandafter\ifx\csname l@#1\endcsname\relax
\typeout{** WARNING: IEEEtran.bst: No hyphenation pattern has been}%
\typeout{** loaded for the language `#1'. Using the pattern for}%
\typeout{** the default language instead.}%
\else
\language=\csname l@#1\endcsname
\fi
#2}}
\providecommand{\BIBdecl}{\relax}
\BIBdecl

\bibitem{6g}
\BIBentryALTinterwordspacing
K.~B. Letaief, W.~Chen, Y.~Shi, J.~Zhang, and Y.~A. Zhang, ``The roadmap to
  6{G} - {AI} empowered wireless networks,'' \emph{CoRR}, vol. abs/1904.11686,
  2019. [Online]. Available: \url{http://arxiv.org/abs/1904.11686}
\BIBentrySTDinterwordspacing

\bibitem{smart}
M.~D. Renzo, M.~Debbah, and D.-T. P.-H. et~al., ``Smart radio environments
  empowered by reconfigurable ai meta-surfaces: An idea whose time has come,''
  \emph{accepted for publication in EURASIP J. Wireless Commun. Netw.}, 2019.

\bibitem{1906-09490}
\BIBentryALTinterwordspacing
E.~Basar, M.~D. Renzo, J.~de~Rosny, M.~Debbah, M.-S. Alouini, and R.~Zhang,
  ``Wireless communications through reconfigurable intelligent surfaces,''
  \emph{CoRR}, vol. abs/1906.09490, 2019. [Online]. Available:
  \url{http://arxiv.org/abs/1906.09490}
\BIBentrySTDinterwordspacing

\bibitem{ch}
C.~Huang, A.~Zappone, G.~C. Alexandropoulos, M.~Debbah, and C.~Yuen,
  ``Reconfigurable intelligent surfaces for energy efficiency in wireless
  communication,'' \emph{accepted for publication in IEEE Trans. Wireless
  Commun.}, 2019.

\bibitem{8466374}
C.~{Liaskos}, S.~{Nie}, A.~{Tsioliaridou}, A.~{Pitsillides}, S.~{Ioannidis},
  and I.~{Akyildiz}, ``A new wireless communication paradigm through
  software-controlled metasurfaces,'' \emph{IEEE Commun. Mag.}, vol.~56, no.~9,
  pp. 162--169, Sep. 2018.

\bibitem{1906-06578}
\BIBentryALTinterwordspacing
Y.-C. Liang, R.~Long, Q.~Zhang, J.~Chen, H.~V. Cheng, and H.~Guo, ``Large
  intelligent surface/antennas {(LISA)}: Making reflective radios smart,''
  \emph{CoRR}, vol. abs/1906.06578, 2019. [Online]. Available:
  \url{http://arxiv.org/abs/1906.06578}
\BIBentrySTDinterwordspacing

\bibitem{1905-00152}
\BIBentryALTinterwordspacing
Q.~Wu and R.~Zhang, ``Towards smart and reconfigurable environment: Intelligent
  reflecting surface aided wireless network,'' \emph{CoRR}, vol.
  abs/1905.00152, 2019. [Online]. Available:
  \url{http://arxiv.org/abs/1905.00152}
\BIBentrySTDinterwordspacing

\bibitem{8647620}
Q.~{Wu} and R.~{Zhang}, ``Intelligent reflecting surface enhanced wireless
  network: Joint active and passive beamforming design,'' in \emph{Proc. of
  IEEE GLOBECOM, Abu Dhabi, UAE}, Dec 2018.

\bibitem{1906-09434}
\BIBentryALTinterwordspacing
M.~Fu, Y.~Zhou, and Y.~Shi, ``Intelligent reflecting surface for downlink
  non-orthogonal multiple access networks,'' \emph{CoRR}, vol. abs/1906.09434,
  2019. [Online]. Available: \url{http://arxiv.org/abs/1906.09434}
\BIBentrySTDinterwordspacing

\bibitem{shen11}
K.~{Shen} and W.~{Yu}, ``Fractional programming for communication
  systems—{P}art {I}: Power control and beamforming,'' \emph{IEEE Trans. on
  Signal Process.}, vol.~66, no.~10, pp. 2616--2630, May 2018.

\bibitem{1904-12475}
\BIBentryALTinterwordspacing
T.~Jiang and Y.~Shi, ``Over-the-air computation via intelligent reflecting
  surfaces,'' \emph{CoRR}, vol. abs/1904.12475, 2019. [Online]. Available:
  \url{http://arxiv.org/abs/1904.12475}
\BIBentrySTDinterwordspacing

\bibitem{1906-02360}
\BIBentryALTinterwordspacing
Q.-U.-A. Nadeem, A.~Kammoun, A.~Chaaban, M.~Debbah, and M.-S. Alouini,
  ``Intelligent reflecting surface assisted multi-user {MISO} communication,''
  \emph{CoRR}, vol. abs/1906.02360, 2019. [Online]. Available:
  \url{http://arxiv.org/abs/1906.02360}
\BIBentrySTDinterwordspacing

\bibitem{DBLP:journals/cm/ShiZCL18}
Y.~Shi, J.~Zhang, W.~Chen, and K.~B. Letaief, ``Generalized sparse and low-rank
  optimization for ultra-dense networks,'' \emph{{IEEE} Commun. Mag.}, vol.~56,
  no.~6, pp. 42--48, 2018.

\bibitem{Gorski2007}
J.~Gorski, F.~Pfeuffer, and K.~Klamroth, ``Biconvex sets and optimization with
  biconvex functions: a survey and extensions,'' \emph{Math. Methods Oper.
  Res.}, vol.~66, no.~3, pp. 373--407, Dec 2007.

\bibitem{IET}
L.~{Subrt} and P.~{Pechac}, ``Intelligent walls as autonomous parts of smart
  indoor environments,'' \emph{IET Commun.}, vol.~6, no.~8, pp. 1004--1010, May
  2012.

\bibitem{convex}
S.~Boydand and L.~Vandenberghe, \emph{Convex Optimization}.\hskip 1em plus
  0.5em minus 0.4em\relax Cambridge University Press, 2004.

\end{thebibliography}
\end{document}